\documentclass[aip,
cha, 
 amsmath,amssymb,
]{revtex4-1}

\usepackage{amsthm}
\usepackage{latexsym, amsmath}
\usepackage[english]{babel}
\usepackage{graphicx}
\usepackage{dcolumn}
\usepackage{bm}
\def\R{{\mathbf R}}
\def\Z{{\mathbf Z}}
\def\C{{\mathbf C}}
\newcommand{\x}{\mathbf{\chi}}

\newtheorem{theorem}{Theorem}[section]

\newtheorem{Definition}[theorem]{Definition}

\newtheorem{proposition}[theorem]{Proposition}
\newtheorem{remark}[theorem]{Remark}

\renewcommand{\figurename}{\small{Figurasss}}
\draft 

\begin{document}


\title{On multistability behavior of Unstable Dissipative Systems}

\author{A.~Anzo-Hern\'andez}
\affiliation{{\textsc{ C\'atedras CONACYT - Benem\'erita Universidad Aut\'onoma de Puebla - Facultad de Ciencias F\'isico-Matem\'aticas},\\}
{Benemerita Universidad Aut\'onoma de Puebla},\\
\textsc{Avenida San Claudio y 18 Sur, Colonia San Manuel, 72570.\\
Puebla, Puebla, M\'exico.
\vskip 2ex}}

\author{H.E.~Gilardi-Vel\'azquez}%

\author{E. Campos-Cant\'on}
 \email{eric.campos@ipicyt.edu.mx}
\affiliation{
 \textsc{Divisi\'on de Matem\'aticas Aplicadas,\\}
 {Instituto Potosino de Investigaci\'on Cient\'{\i}fica y Tecnol\'ogica A.C.}\\
 \textsc{Camino a la Presa San Jos\'e 2055 col. Lomas 4a Secci\'on, 78216,\\
 San Luis Potos\'{\i}, SLP, M\'{e}xico \vskip 2ex}}

\begin{abstract}
We present a dissipative system with unstable dynamics called unstable dissipative system which are capable of generating a multi-stable behavior, \textit{i.e.}, depending on its initial condition the trajectory of the system converge to a specific basin of attraction. A piecewise linear (PWL) systems is generated based on unstable dissipative systems  (UDS) whose main attribute when they are switched is the generation of chaotic trajectories with multiple wings or scrolls. For this PWL system a structure is proposed  where both the linear part and the switching function depend on two parameters.  We show the range of values of such parameters where the PWL system presents a multistable behavior and  trajectories with multiscrolls.
\end{abstract}

\maketitle

\section{Introduction}

In the evolution of a complex system, there are several possible (coexisting) basins of attraction with a sink that traps the system trajectory depending on its initial state. This phenomenon is usually called multistability and appears in a wide variety complex systems \cite{Campos-Canton2010a, Angeli2007}. The correct interpretation of a sink depends on the complex system being studied. For instance, in the context of biology there are many examples of systems that manifest multistability phenomena. An  example worth mentioning is the cellular differentiation in order to understand human development and distinct forms of diseases. Here, multistability is understood as a processes in which a gene regulation network alternates along several possible types of  cell \cite{Ghaffarizadeh2014}. Another example comes from the nonlinear chemical dynamics, where multistability is understood as the different possible  final chemical states \cite{Sagues2003}. In this context, the archetype system is the Oregonator oscillator, where concentrations of the reacting species oscillate between two stable final states (bistability).  Several examples can be cited ranging from medicine \cite{Haddad2011}, electronic \cite{Patel2014}, visual perception \cite{Gershman2012}, superconducting \cite{Jung2014}, etc.  All of these examples motivate the current research works that address the challenger posed by R. Vilela Mendes in  \cite{Mendes2000} of -identifying the universal mechanism that leads to multistability and to prove rigorously under what circumstances the phenomenon  may occur.  

One feasible mode to address this challenge is through the formalism of dynamical systems where the concepts of basin of attraction, stability, convergence among others have a mathematical definition and also allow us to use some tools from stability theory to analyze its behavior.  It is worth to note that this situation is similar with the research works some year ago where chaotic behavior was modeled and interpreted from the point of view of dynamical systems. Since then, various dynamical systems with a chaotic behavior have been proposed (some examples are the Lorenz, Chua and  R{\"o}ssler systems, to name a few).

In the context of dynamical systems, an attractor is defined as a subset of the phase space toward which the trajectories of the dynamical systems converge to it. Attractors can be fixed points, limit cycles, quasiperiodic, chaotic or hyper-chaotic orbits. The basin of attraction is defined as the set of all the initial conditions in the phase space whose corresponding trajectories converge to an attractor \cite{Kengne2017,Giesl2007}. Concepts of convergent trajectories and attractor stability are usually associated with a energy-like term called Lyapunov function. Then, with the above concepts it can be said that a multistable dynamical system is a dynamical system that, depending on its initial condition, its trajectories can converge to two or more mutually exclusive Lyapunov stable attractors \cite{Haddad2011} .

Some formal definitions of multistable behavior have been proposed by D. Angeli in \cite{Angeli2007} and Q. Hui in \cite{Hui2014} for discontinuous dynamical systems. On the other hand, some methodologies to induce a multistable behavior by coupling two o more dynamical systems have been reported. For example, E. Jim\'enez-L\'opez \textit{et. al.} have generated multistable behavior by employing a pair of Unstable Dissipative Systems (UDS) of Type I, coupled in a master-slave configuration \cite{Jimenez2013}. It is worth to mention that an UDS is a Piecewise Linear System (PWL) which is classificated  in two types according to the eigenvalues of the linear operator. On the other hand, C.R. Hens \textit{et.al.} have shown that two coupled R{\"o}ssler oscillators can achieve a certain type of multistability called extreme, where the number of coexisting attractors is infinite \cite{Hens2012}.  It has also been observed that by an appropriate modification of the equations, some classical chaotic systems can exhibit also a multistable behavior. For example,  in \cite{Kengne2017,Kengne2016} J. Kengne \textit{et.al.} have proposed a system based on the Duffing-Holmes system and Chua's oscillator. They have shown that in a given range of its parameter values this system exhibit coexisting attractors. Additionally, hand, H.E. Gilardi-Vel\'azquez \textit{et. al.} introduce a multistable system generated with a Piecewise Linear (PWL) system based on the Jerk equation, in which the switching among the different phase-space regions is driven by means of the Nearest Integer or the round(x) function, such that the system display  infinite attractors along one dimension \cite{Gilardi-velazquez2017}. The experimental evidence of multistability for the R\"ossler oscillator have been reported by M. Patel \textit{et al} in \cite{Patel2014}. On the other hand, C. Li in \cite{Li2013} and D.Z.T. Njitacke \textit{et al} in \cite{Njitacke2016}, have been observed that multistable behavior is also presented in the Butterfly-Flow system and in the memristive diode bridge-based Jerk circuit, respectively. It is worth to mention that for some discrete-time chaotic systems, the multistable behavior is also displayed  \cite{Carvalho2001a,Astakhov2001}

In this paper we propose two methodologies to generate a multi stable behavior in a UDS of Type I and Type II. The first methodology consist in introduce a bifurcation parameter in the linear operator of the UDS of Type I. With such parameter, we can change the location of the stable and unstable manifold until the trajectories are trapped in a specific region of the space. In regard to our second methodology, we consider a UDS Type II and modify the switching law without changing the linear operator. With both methodologies we can design a priory the number of multistable regions by introducing another switching surfaces. 

This paper is organized as follows: in Section 2 we propose a definition of a multistable dynamical systems.  In section 3, we define and describe the main features of an UDS. Even we present in this section the conditions under which a dynamical system is an UDS Type I or Type II system. In section 4 we present in detail our proposed methodology to induce multistability in a UDS Type I system and in section 5 the corresponding methodology for UDS Type II. In section 6 we present some concluding remarks.

\section{Multistable dynamical system}

Let $\R^n$ be a metric space with Euclidean metric $d$. The dynamical system on the metric space $\R^n$  that we consider in this paper is an autonomous nonhomogeneous first order lineal ordinary differential equation system of the form:
\begin{equation}\label{eq:affine}
 \dot{\x} = f(\x)=A\x + g(\x), \quad \x (0) = \x_{o},
\end{equation}

\noindent where $\x \in \mathbf{R}^{n}$ is the state vector,  $A = \{a_{ij}\}^{n}_{i,j = 1} \in \mathbf{R}^{n \times n}$ is a non-singular linear operator with $a_{ij} \in \R$ and;  $g: \mathbf{R}^{n} \rightarrow \mathbf{R}^{n}$ is a piecewise constant vector which commutes as follows:
\begin{equation}\label{eq:gfunction}
g(\x) =  \left\{
\begin{array}{lll}
  B_{1}, & \text {if }    & \x \in S_{1} = \{\x \in \mathbf{R}^{n}:  G_{1}(\x) < \delta_{1} \}; \\
  B_{2}, & \text {if }    & \x \in S_{2}  = \{\x \in \mathbf{R}^{n}:  \delta_{1}  \leq G_{2}(\x) < \delta_{2} \}; \\
     \vdots    &       &  \vdots \\
   B_{m}, & \text {if }   & \x \in S_{m} = \{\x \in \mathbf{R}^{n}:  \delta_{m-1}  \leq G_{m}(\x)  \}; \\
\end{array} \right.
\end{equation}

\noindent where $B_{i} = (b_{i1},\ldots,b_{in})^T \in \R^{n}$, for $i=1,\ldots,m$ , are vectors with real entries; and $S = \{ S_{1}, S_{2}, \ldots, S_{m}\}$ is a finite partition  of the phase space, which satisfy $\R^{n}=\bigcup_{i=1}^ m S_i$ and $\bigcap_{i=1}^mS_i = \emptyset$.  Each domain $S_{i}$ is defined by surfaces $\Sigma_i$ in terms of $\delta_{i}$ (with $1\leq i \leq m-1$) that act as a separatrices (or boundaries) between two consecutive switching domain. In what follows we call each $\Sigma_i$ the switching surfaces. Furthermore, we assume that each set $S_{i}$ has a saddle equilibrium point $\x_i^{*}\in S_{i}$. If $\lambda_{j} = \alpha_{j} + i \beta_{j}$ is  either a complex eigenvalue $\beta_j\neq 0$ or real eigenvalue  $\beta_j =0$  of the linear operator $A$ and $\bar{\textit{v}_{j}} \in \mathbf{R}^{n}$ its corresponding eigenvector, then the stable set is $E^{s} = \textit{Span}\{ \bar{\textit{v}_{j}} \in \mathbf{R}^{n} : \alpha_{j} < 0 \}$ and the unstable set $E^{u} = \textit{Span}\{ \bar{\textit{v}_{j}} \in \mathbf{R}^{n} : \alpha_{j} > 0 \}$ \cite{Guzzo2010}.   

Let $\phi_{t}(\x_0) \in \mathbf{R}^{n}$ be the solution curve or trajectory of  \eqref{eq:affine} given the initial condition $\x_0$. Thus $\phi:\R^n\times\R\to \R^n$: 
\begin{Definition}\label{def1} 
A closed invariant set $\mathcal{A}\subseteq\R^n$ is called an attracting set of  \eqref{eq:affine} with flow $\phi_t$, if there exist a neighborhood  $U\subseteq\R^n$ of $\mathcal{A}$ with $\phi_t(U)\subseteq U$ and $\mathcal{A}\subseteq U$ such that  
$$\mathcal{A}=\bigcap_{t=0}^\infty \phi_t(U),$$
where $$\phi_t(U)=\{\phi_t(x)|x\in U \}.$$  An attractor  of  \eqref{eq:affine} is an attracting set which contains a dense orbit.
\end{Definition}

\begin{Definition}\label{defBasinAttra} The basin of attraction of  $\mathcal{A}$ is the set of initial conditions whose trajectories converge to the attractor, that is $U=\Omega(\mathcal{A}) = \{ \x_{0} \in \R^{n}: \phi_{t}(\x_{0}) \to \mathcal{A}  \quad \text{as} \quad  t \to \infty   \}$. 
\end{Definition}

Thus, an attractor is a closed invariant set $\mathcal{A}$ and there is an open neighborhood $U \supset \mathcal{A}$ such that the trajectory $\x(t)=\phi_{t}(\x_0)$ of any point $\x_0 \in U$ satisfies $d(\phi_{t}(\x_0),\mathcal{A})  \to 0$ as $t \to \infty$; where $d(\x,\mathcal{A}) =\text{ inf}\{d(\x,x_0)| \; \x\in\phi_t(\x_0) \;\text{and}\; x_{0} \in \mathcal{A} \}$.

There are different types of attractors, i.e., stable equilibrium point, limit cycle, a set generated by a chaotic trajectory. Based on the aforementioned definition, we will be considering the following definition of a  generalized multistable system throughout this work:

\begin{Definition}\label{def3}
We say that the dynamical system given by \eqref{eq:affine} is generalized multistable if there exist more than one basin of attraction, i. e. $\Omega(\mathcal{A}_1), \ldots, \Omega(\mathcal{A}_k)$, with $2\leq k\in \Z$.
\end{Definition}
In the context of generalized multistability, the coexistence of multiple attractors $\mathcal{A}_1,\ldots, \mathcal{A}_k$ makes the distance $d(\phi_{t}(\x_0),\mathcal{A}_i)$ takes different values that depends on the initial condition $\x_0$. For instance, the distance $d(\phi_{t}(\x_0),\mathcal{A}_i)=0$ if  $\x_0\in \Omega(\mathcal{A}_i)$, but $d(\phi_{t}(\x_0),\mathcal{A}_i)\neq0$ if   $\x_0\in \Omega(\mathcal{A}_j)$, with $i\neq j$.

\begin{remark}
It is important to characterize different types of multistability as follows:
\begin{enumerate}
\item[1.-] There exists a set $\{\x_{i}^{*}\}^{m}_{i=1}$ of saddle equilibrium points of  \eqref{eq:affine} in $\mathbf{R}^{n}$. The basin of attraction of each equilibria is given by the stables set $E^s$.  This type of multistable states is known as multistability.  
\item[2.-] Due to the phase space $\R^n$ is partitioned in a finite number $m\in\Z$ of domains $S_{i}$  and each equilibrium point is located at $\x^{*}_{i} = A^{-1}B_{i}\in S_i\subseteq\R^n$, for $i = 1,\ldots,m$. So the basins of attraction of  $\x^{*}_{i}$ is determined by the stable set $E_i^s$ restricted to $S_i$.
\item[3.-] When the trajectory does not converge to the equilibria, instead oscillates around them and exist at least two basin of attractors $\Omega_{i}=\Omega(\mathcal{A}_i)$ and $\Omega_{j}=\Omega(\mathcal{A}_j)$, this type of multistable states is known as generalized multistability. This is the target of this work.
\end{enumerate}
\end{remark}

\section{Chaotic attractors based on Unstable Dissipative Systems}\label{sec_UDS}

We consider the following family of affine linear systems:
\begin{equation}\label{eq:PWL}
\dot{\x} = A\x + B(\x),
\end{equation}

\noindent where $\x= (x_{1},x_{2},x_ {3})^{\top} \in \mathbf{R^{3}}$ is the state vector,  the real matrix $A  \in \mathbf{R}^{3 \times 3}$  is a non-singular linear operator; and $B: \mathbf{R}^{3} \rightarrow \mathbf{R}^{3}$ is  a piecewise constant vector which is controlled by  a step function  accoding to the domain $S_i$. For simplicity and without loss of generalization, we assume that \eqref{eq:PWL} is given by the jerk type equation \cite{Campos-Canton2012a}:
\begin{equation}\label{eq:uds}
   A = \left( \begin{array}{ccc}
          0 & 1 & 0 \\
          0 & 0 & 1 \\
         -\alpha & -\beta & -\gamma \end{array} \right), \quad
  B(\x) =  \left( \begin{array}{c}
           0 \\
           0 \\
            \sigma(\x)   
          \end{array} \right);
\end{equation} 

\noindent where $\alpha, \beta, \gamma \in \mathbf{R}$ and $\sigma(\x): \mathbf{R}^{3} \rightarrow \mathbf{R}$ is a step function  which is determined by a switching law to control the equilibria of the system, as follows:
\begin{equation}\label{eq:switching_law}
\sigma(\x) =  \left\{
\begin{array}{lll}
  b_{1}, & \text {if}    & \x \in S_{1} = \{\x \in \mathbf{R}^{3}:   \mathbf{v}^{\top}\x < \delta_{1} \}; \\
   b_{2}, & \text {if}    & \x \in S_{2}  = \{\x \in \mathbf{R}^{3}:  \delta_{1}  \leq \mathbf{v}^{\top}\x < \delta_{2} \}; \\
     \vdots    &       &  \vdots \\
   b_{m}, & \text {if}   & \x \in S_{m} = \{\x \in \mathbf{R}^{3}:  \delta_{m-1}  \leq \mathbf{v}^{\top} \x \};\\
\end{array} \right.
\end{equation}

\noindent with $b_{i} \in \mathbf{R}$ and $S_{i}$, $i = 1,\ldots,m$, generates a partition of the phase space, with $\mathbf{v} \in \mathbf{R}^{3}$ (with $\mathbf{v} \neq 0$)  a constant vector and $\delta_1 \leq \delta_2 \leq \cdots \leq \delta_{m-1} $ determine switching surfaces $\Sigma_j=\{ \x\in\R^3|\mathbf{v}^{\top}\x=\delta_j\}$, $j=1,\ldots,m-1$. In particular, we assume that switching surfaces $\Sigma_j$ are defined by using  $\mathbf{v} = [1,0,0]^{\top} \in \mathbf{R}^{3}$ and different values of $\delta_j$. The role of the switching function $\sigma$ is to specify which constant vector is active at a given domain $S_i$, that is, if $\sigma(\x) = b_{i}$ for $i \in I = \{1,\ldots,m\}$, then the affine linear system that governs the dynamics in the switching domain $S_{k}$ is given by $\dot{\x} = A\x + (0,0, b_{i})^{\top}$. 

Our case study is when each switching domain contains an single saddle equilibrium point  located at $\x_{i}^{*} = A^{-1}B_{i}$, with $i \in I$. The idea is to generate different basins of attraction $\Omega_j=\Omega(\mathcal{A}_j)$  such that for any initial condition $\x_{0} \in \Omega=\bigcup_{j=1}^ k \Omega_{j} \subset \mathbf{R}^{3}$, the trajectory $\phi(\chi_{0})$ of the system \eqref{eq:PWL}-\eqref{eq:uds}  converges at only one attractor $\mathcal{A}_{j}$. We are considering generalized multistability so the trajectory needs to remain oscillating chaotically. We start considering only one basin of attraction of a  multiscroll chaotic attractor. The mechanism of generation of multiscroll attractors based on this class of systems is due to the stable and unstable manifolds. For example, considering two domains $S_i$ and $S_{i+1}$, and the commutation surface $\Sigma_{i}$ between them. When the trajectory $\phi_t(\x_0)$, with  initial condition $\x_0\in S_i\cap \Omega$, reaches to the commutation surface $\Sigma_{i}$ and crosses to the region $S_{i+1}$, where it is again trapped in a new scroll with equilibrium point  $\x_{i+1}^{*} = A^{-1}B_{i+1}$. 
There are two important facts about the generation of multiscroll attractors, first that the scrolls are generated due to the complex conjugate eigenvalues with positive real part, so the scroll increasing their size due to the unstable manifold. Second, that the trajectory of the system oscillating around the equilibrium point $\x^{*}_{i}$ escapes from the domain $\mathcal{S}_{i}$. This occurs near the unstable manifold $E_i^{u}\subset \mathcal{S}_{i}$ where it crosses the commutation surface  and  it is attracted by the stable manifold $E_{i+1}^s\subset \mathcal{S}_{i+1}$ towards the equilibrium point  $\x^{*}_{i+1}$  in the domain $\mathcal{S}_{i+1}$. The process is repeated in the inverse way forming scrolls around each equilibrium point.
In this context, the system \eqref{eq:PWL}-\eqref{eq:uds} can display various multi-scroll attractors as a result of a combination of several unstable one-spiral trajectories \cite{Ontanon-Garcia2014}, where the switching between regions is governed by the switching function \eqref{eq:switching_law}.

In what follows, we assume that the eigenspectra $\Lambda = \{ \lambda_{1},\lambda_{2},\lambda_{3}\}$ of the linear operator $A \in \R^{3\times 3}$ has the following features: a) at least one eigenvalue is a real number  and; b) at least two eigenvalues are complex conjugate numbers. Furthermore, we consider that the sum of eigenvalues of $\Lambda$ is negative. A dynamical system defined by the linear part of the system \eqref{eq:PWL} that satisfy the above requirements is called an Unstable Dissipative System (UDS) \cite{Campos-Canton2010a}. 

\begin{Definition}\label{def_uds} Let $\Lambda = \{ \lambda_{1},\lambda_{2},\lambda_{3}\}$ be the eigenspectra of the lineal operator $A\in \mathbf{R^{3\times 3}}$, such that $\sum^{3}_{i=1} \lambda_{i}  < 0$, with $\lambda_{1}$ a real number and $\lambda_{2},\lambda_{3}$ two complex conjugate numbers. A system given by the linear part of the system \eqref{eq:PWL} is said to be a UDS \textit{Type I} if $\lambda_{1}<0$ and $ \textit{Re}\{\lambda_{2,3}\}>0$; and it is  \textit{Type II} if $\lambda_{1}>0$  and $ \textit{Re}\{\lambda_{2,3}\}<0$.
\end{Definition}

The above definition implies that the UDS \textit{Type I} is dissipative in one of its components but unstable in the other two, which are oscillatory. The converse is the UDS \textit{Type II}, which are dissipative and oscillatory in two of its components but unstable in the other one.  The following result ( based on the results  in \cite{Campos-Canton2012a}) provide conditions to guaranteed that the system \eqref{eq:PWL} is \textit{UDS} \textit{Type I} or \textit{Type II} for a general lineal operator $A = \{ \alpha_{ij} \} \in \R^{3}$, with $\alpha_{ij} \in \R$ for $i,j=1,2,3$.

\begin{proposition}\label{prop1}
 Consider the family of affine lineal systems \eqref{eq:PWL} with lineal operator $A$ given by \eqref{eq:uds} with $\alpha, \beta, \gamma \in \mathbf{R}$. Let $\{a,b,c\}$ be a set of non zero real numbers called control parameters. If $\alpha = c(a^{2} + b )$, $\beta = a^{2} + b + 2ac$ and $\gamma = c-2a $ with $b,c>0$ and $a<c/2$, then the system  \eqref{eq:PWL}-\eqref{eq:uds} is based on UDS \textit{Type I}; on the other hand, if $b>0$ and $a,c<0$, then the system is based on UDS \textit{Type II}.
\end{proposition}

\begin{proof}
The characteristic polynomial of the lineal operator \eqref{eq:uds} is:
\[
\begin{array} {lcl} 
p(\lambda) & = & \lambda^{3} + \gamma  \lambda^{2} + \beta \lambda + \alpha, \\          
                  & = & \lambda^{3} + (c-2a) \lambda^{2} + (a^{2} + b + 2ac) \lambda + (ca^{2} + cb ) ,\\
                  & = & (\lambda + c)(\lambda^{2} - 2a\lambda + (a^{2} + b)) .
\end{array}\]

\noindent The roots of $p(\lambda)$ give the following expressions for the eigenspectra
$\Lambda = \{ \lambda_{1},\lambda_{2},\lambda_{3}\}$  of \eqref{eq:uds}: $\lambda_{1} = -c$ and $\lambda_{2,3} = a \pm i\sqrt{b}$. Note that $\lambda_{1}<0$ and  $\sum^{3}_{i=1}  \lambda_{i}  = -c + 2a< 0$ if $a<c/2$ and  $c>0$.  Then, according to Definition \eqref{def_uds} the system \eqref{eq:PWL}-\eqref{eq:uds} is UDS \textit{Type I}. On the other hand, if $a,c<0$, then $\lambda_{1}>0$ and the above summation is still negative since $a>c/2$, which implies that the system is  UDS \textit{Type II}.
\end{proof}
\begin{proposition}\label{prop2}
 Consider the family of affine lineal systems given by  \eqref{eq:PWL}, the lineal operator $A$ based on the jerk system \eqref{eq:uds} with $\alpha, \beta, \gamma \in \mathbf{R}$. If $\alpha > 0$, $0<\beta < \alpha / \gamma$ and $\gamma >0 $, then the system  \eqref{eq:PWL}-\eqref{eq:uds} is based on UDS \textit{Type I}.
\end{proposition}

\begin{proof} Suppose $\alpha, \gamma >0 $. Since, by definition, $-\gamma = Trace(A)=\sum_{i=1}^3 \lambda_i<0$, system \eqref{eq:PWL} is dissipative. Additionally,
with $\alpha=det(A)$ the  system \eqref{eq:PWL} has a saddle equilibrium,
which is determined by the characteristic polynomial of the lineal operator \eqref{eq:uds} is:
$$p(\lambda) =  \lambda^{3} + \gamma  \lambda^{2} + \beta \lambda + \alpha, $$

\noindent which for  $\beta < \alpha / \gamma$, according with Hurwitz polynomial criterion, implies unstability. Due to $\alpha$, $\beta$ and $\gamma$ are positive and according to Descartes' rule of signs  the characteristic polinomial has no positive roots, so it  has only one negative root due to the equilibrium point is a saddle. Then the eigenspectra is given by one negative real eigenvalue and a pair of complex conjugate eigenvalues with positive real part.   
\end{proof}
Only one  UDS type I  generates an unstable spiral around the equilibrium point $\x^*$ such that a trajectory $\phi_t(\x_0)$, $\x_0\in \R^3- E^s$,crosses many times a Poincar\'e plane $\Pi$  if  $\Pi$ is deined at the equilibrium point $\x^*$ and perpendicular to the unstable manifold $E^u \perp \Pi$. Let $\{\x_{i}^{*}\}^{m}_{i=1}$ be a set of  saddle equilibria of the PWL system \eqref{eq:affine} based on UDS Type I
\begin{figure}[t]
\centering
\hspace{-25pt}\includegraphics[width=14cm,height=8cm]{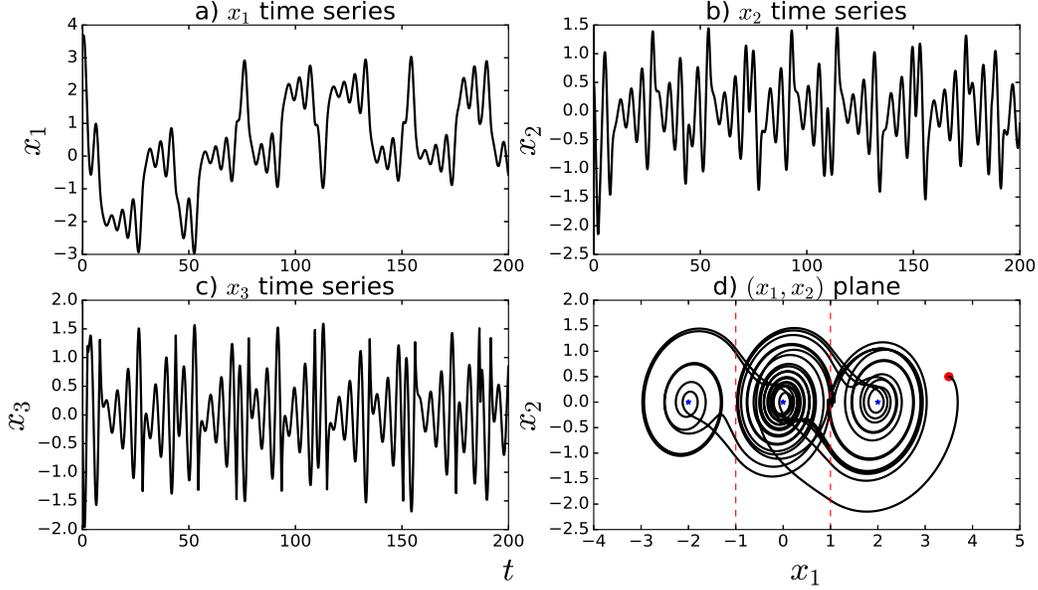}
\caption{a) $x_1$ time series. b) $x_2$ time series. c) $x_3$ time series. d) Projection of the attractor onto the ($x_{1},x_{2}$) plane based on UDS \textit{Type I} with control parameters $a = 0.12501$, $b = 1.5625$ and  $c = 1.25$; and switching law \eqref{eq:switching_law_6} . The dashed lines mark the division between the switching surfaces and the red dot indicates the initial position at $\chi_{0} = (3.5, 0.5, 0)^{\top}$. }
\label{fig:Fig1}
\end{figure}

\begin{figure}[t]
\centering
\hspace{-25pt}\includegraphics[width=14cm,height=8cm]{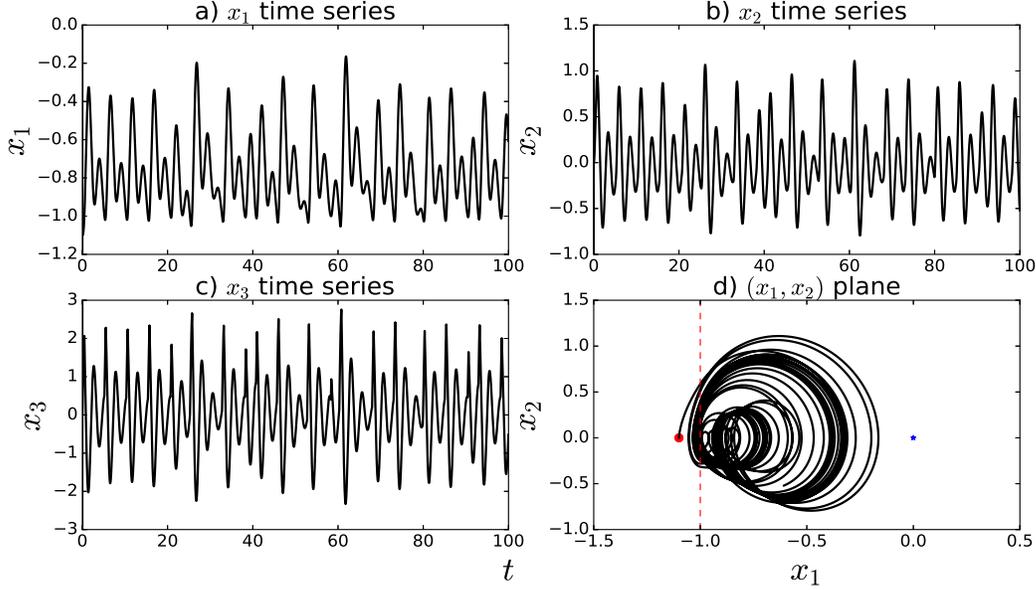}
\caption{a) $x_1$ time series. b) $x_2$ time series. c) $x_3$ time series. d)Projection of the attractor onto the ($x_{1},x_{2}$) plane based on UDS \textit{Type II}, with control parameters $a = -0.3494$, $b = 5.9469$ and  $c = -0.0988$; and switching law \eqref{eq:switching_law_6} . The dashed lines mark the division between the switching surfaces and the red dot indicates the initial position at $\chi_{0} = (-1.1, 0, 0)^{\top}$. }
\label{fig:Fig2}
\end{figure}

\begin{Definition}\label{def_multiscroll} 
Let $\{\x_{i}^{*}\}^{m}_{i=1}$ be a set of  equilibria of a PWL system \eqref{eq:affine} based on UDS Type I that generates a chaotic attractor $\mathcal{A}$. We say that the system \eqref{eq:affine} generates a multiscroll chaotic attractor with the minimum of saddle equilibria, if the chaotic trajectory $\phi_t(\x_0)\subset S_i$, $\x_0\in \Omega(\mathcal{A})$, crosses each Poincar\'e plane $\Pi_i$ defined at  $\x_{i}^{*}$ and $\Pi_i\perp E_i^u$ more than once  before leaving the domain $S_i$ and goes to domains $S_{i-1}$ or $S_{i+1}$. 
\end{Definition}
\noindent {\bf Example 1:} 
In order to illustrate the approach, we consider the  dynamical system defined by \eqref{eq:PWL} with control parameters $a = 0.12501$, $b = 1.5625$ and  $c = 1.25$. Then, according to the Proposition \eqref{prop1}, the last row of the lineal operator \eqref{eq:uds} is defined by the following elements: $\alpha = 1.9727$, $\beta = 1.2656$ and $\gamma = 1$. With this selection of control parameters, the eigenvalues of $A$ are  $\lambda_{1} = -1.25$ and $\lambda_{2,3} = 0.125  \pm \textit{i} 1.25 $, which according to Definition \ref{def_uds}, the system is an UDS \textit{Type I} . We define the switching law as: 

\begin{equation}\label{eq:switching_law_3}
\sigma(\x) =  \left\{
\begin{array}{lll}
   2,  & \text{if}  & \x \in S_{1} = \{\x \in \mathbf{R}^{3}: x_{1} > 1\}; \\
   0,  & \text{if}  & \x \in S_{2} = \{\x \in \mathbf{R}^{3}: -1 < x_{1} \leq 1\};  \\
  -2, & \text{if}  & \x \in S_{3} = \{\x \in \mathbf{R}^{3}:  x_{1} \leq -1\}.   \\
\end{array} \right.
\end{equation} 

Then, the equilibria for this system are located at $\chi^{*}_{1} =  (2,0,0)^T$, $\chi^{*}_{2} =  (0,0,0)^T$ and $\chi^{*}_{3} =  (-2,0,0)^T$. Figure \eqref{fig:Fig1} shows the time series of the state variables a) $x_{1}$, b) $x_{2}$, c) $x_{3}$ and d) the projection of the attractor based on UDS onto the $(x_1,x_2)$ plane, where we use the switching law \eqref{eq:switching_law_3} and initial condition $\chi_{0} = (3.5, 0.5, 0)^{\top}$. It is worth to note that Definition \eqref{def_multiscroll} is satisfied.

\noindent {\bf Example 2:} 
As a second example, we consider the following set of parameters $\{ a = -0.3494, b = 5.9469, c = -0.0988 \}$. Then, according to the results of the Proposition \eqref{prop1}, $\alpha = -0.6$, $\beta=6$ and $\gamma = 0.6$. With this selection of parameters the eigenvalues of $A$ are  $\lambda_{1} = 0.0989$ and $\lambda_{2,3} = -0.3494  \pm  2.4386 \textit{i}$, which according to Definition \ref{def_uds}, the system is an UDS \textit{Type II} . In particular, for this second example we  define the following switching law:
 
\begin{equation}\label{eq:switching_law_3_prima}
\sigma^{*}(\x) =  \left\{
\begin{array}{lll}
    0,  & \text{if}  & \x \in S_{1} = \{\x \in \mathbf{R}^{3}: -1 \leq x_{1}  \};  \\
    7,  & \text{if}  & \x \in S_{2} = \{\x \in \mathbf{R}^{3}:  x_{1} <  -1 \}.   \\\end{array} \right .
\end{equation} 
Then, the equilibria for this system are located at $\chi^{*}_{1} =  (0,0,0)^T$ and $\chi^{*}_{2} =  (-11.6667,0,0)^T$ . The unstable manifolds $E_1^u$ and $E^u_2$ lead the trajectory $\phi_t(\x_0)$ toward the switching surface $\Sigma_1=\{\x\in\R^3|x_1=-1\}$ and $\x_0\in\Omega$. The basin of attraction $\Omega$ is between the stable manifols $E^s_1$ and $E^s_2$. In Figure \eqref{fig:Fig2} we illustrate the  dynamics of a switching system based on UDS-\textit{Type II} for the initial conditions $\chi_{0} = (-1.1, 0, 0)^{\top}$.

\section{Emerging multistability in a multiscroll attractor based on UDS \textit{Type I}}\label{sec_multistabilityUDSI}

Based on the previous description of a UDS, in this section we consider the PWL system given by  \eqref{eq:PWL}-\eqref{eq:uds}. The goal is to introduce a bifurcation parameter $k$ to this kind of systems in order to go from monostability to multistability. The modification to the system needs to satifiy the folowing requierement: The equilibria of the system  need not depend on either the parameter $k$ or parameters $\alpha$, $\beta$ and $\gamma$ of the linear operator $A$.  This lets the parameter $k$ control the manifolds $E^s$ and $E^u$ in each $S_i$ in order to trap the trajectory in only a single-scroll attractor. The candidate to be our linear operator is given in the following way:
$$ A = \left( \begin{array}{ccc}
          0 & 1 & 0 \\
          0 & 0 & 1 \\
         -k \alpha & -k \beta & -k \gamma \end{array} \right). $$ 
We have two problems, the forme is that the paramete $k$ can modify the dissipativity of the system given by $-k\gamma$, and the second is that the equilibria  are given  by $\x^*=(\sigma(\x)/k \alpha,0,0)^\top$. The dissipativity of the system is kept if $-k \gamma$ is arbitrarily set to $-1$.
We require to preserve the location of the equilibria of the PWL system so we need to multiply  $\sigma(\x)$ by $k\alpha$. Thus the linear operator $A$ and vector $B$ of the PWL system  \eqref{eq:PWL}   are given as follows:

\begin{equation}\label{eq:uds_type1}
   A = \left( \begin{array}{ccc}
          0 & 1 & 0 \\
          0 & 0 & 1 \\
         -k\alpha & -k\beta & -1 \end{array} \right), \quad
  B(\x) =  \left( \begin{array}{c}
           0 \\
           0 \\
           k\alpha \sigma(\x)   
          \end{array} \right);
\end{equation} 

\noindent where $k \in \mathbf{R}^{+}$, $\alpha = 1.9727$ and $\beta = 1.2656$.The $\sigma(\x)$ function is given by \eqref{eq:switching_law_3}.
For $k=1$, we have the particular case given in the example 1 of the Section \ref{sec_UDS}, the system given by \eqref{eq:PWL} and \eqref{eq:uds_type1} satisfies the requirements of Definition \eqref{def_uds}. In Figure \eqref{fig:Fig1} d) we can observe a triple-scroll attractor  $\mathcal{A}$. There is only one basin of attraction $\Omega(\mathcal{A})\subset \mathbf{R^{3}}$ that contains a chaotic attractor.  In this section we use the parameter $k$  as a bifurcation parameter whose role is to modify the location of the stable $E^{s}$ and unstable $E^{u}$ manifolds. In this sense, $k$ change the dynamical behavior of the  UDS \textit{Type I} system from mono-stable to multi-stable, {\it i.e.}, from multiscrol attractorl to three different single-scroll attractors. It is worth to  note that by changing $k$, the switching surfaces and the equilibria remain unchanged.  If we increase the value of  $k$ parameter, the manifold directions change in such a way that for a given initial condition, the trajectory can not display a triple-scrolls attractor as before.

\begin{figure}[t]
\centering
\hspace{-25pt}\includegraphics[width=14cm,height=8cm]{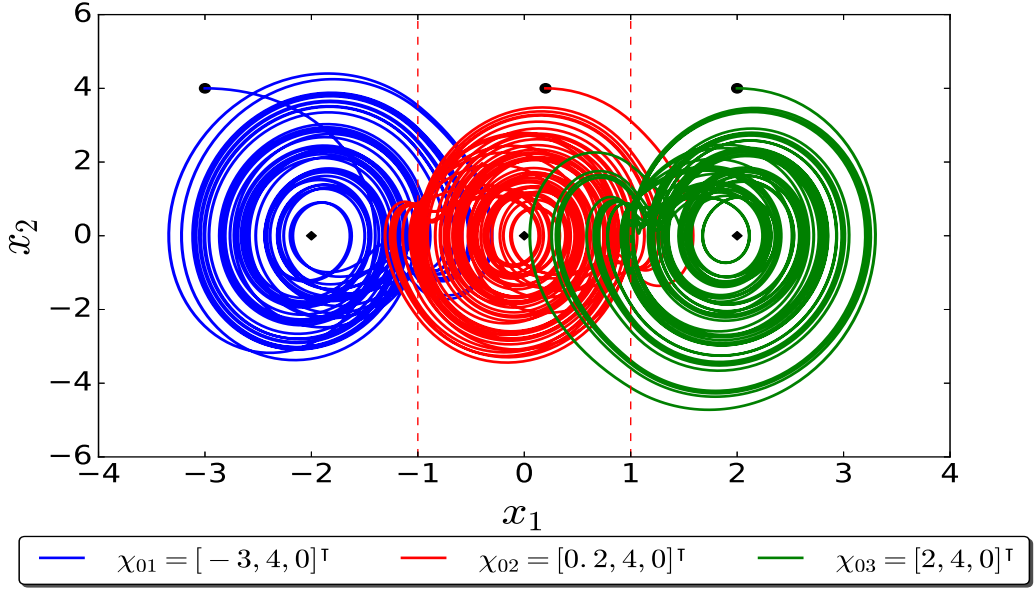}
\caption{Projections of three different attractors onto the ($x_{1},x_{2}$) plane based on UDS \textit{Type I}, with $\alpha = 1.9727$ and $\beta = 1.2656$, switching law \eqref{eq:switching_law_3} and $k = 7$ . The black dots indicates the initial conditions.} 
\label{fig:Fig3}
\end{figure}

\noindent {\bf Example 3:} 
For $k=7$  we have $\hat\alpha=k\alpha= 13.8089$, $\hat\beta=k\beta=8.8592$ and $\hat \gamma=k\gamma=1$. These new parameters $\hat\alpha$, $\hat \beta$ and $\hat\gamma$ satifies proposition \ref{prop2}, so the PWL system is based on UDS type I. Now, the trajectory is trapped in a single-scroll attractor, insead of a multiscroll attractor. There are three basins of attractions $\Omega_i$, with $i=1,2,3$. The attrators generated in a multistable state is shown in Figure  \eqref{fig:Fig3}, where  the following three distinct initial condition were used: 
 $\chi_{01} = (-3, 4, 0)^{\top} \in \Omega_{3} $, $\chi_{02} = (0.2, 4, 0)^{\top} \in \Omega_{2}$ and $\chi_{03} = (2,4, 0)^{\top} \in \Omega_{1}$. 
 
\begin{figure}[t]
\centering
\hspace{-25pt}\includegraphics[width=14cm,height=8cm]{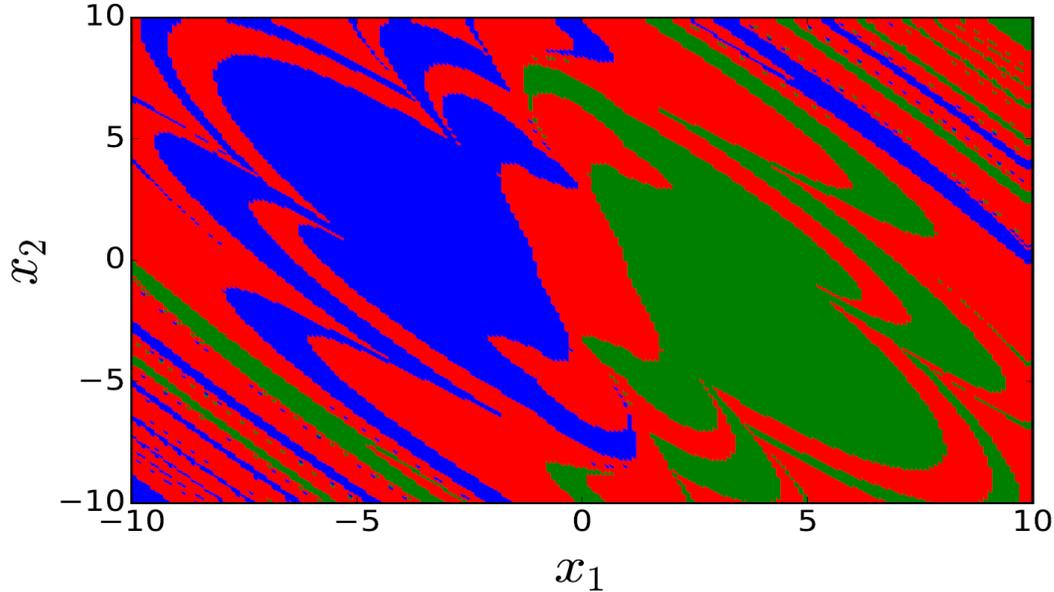}
\caption{Three basins of attraction $\Omega(\mathcal{A}_{i})$ (for $i=1,2,3$) generated with the switching law \eqref{eq:switching_law_3} in $x_{1} \in [-10,10]$, $x_{2} \in [-10,10]$ and $x_{3}=0$. Green dots are used for initial conditions that becomes trapped in $S_{1}$, red dots for $S_{2}$ and blue dots for $S_{3}$.} 
\label{fig:Fig4}
\end{figure}

Next, we vary the initial condition of the dynamical system defined by \eqref{eq:PWL}-\eqref{eq:uds_type1} based on  UDS \textit{Type I} on the $(x_{1},x_{2}$) plane in order to identify  shapes of different basins of attraction $\Omega_i$ of each attractor $\mathcal{A}_{i}$ (for $i=1,2,3$) . In Figure \eqref{fig:Fig4} we show a section of three basins of attractions when the states $x_{1}$ and $x_2$ are varied from $-10$ to $10$  and $x_{3}=0$.  

A green dot means that for such initial condition, the system is trapped in the attractor $\mathcal{A}_{1}$. In similar way, the red dot correspond to the basin of attraction $\mathcal{A}_{2}$ and blue dots to $\mathcal{A}_{3}$.

\begin{figure}
\centering
\hspace{-25pt}\includegraphics[width=14cm,height=8cm]{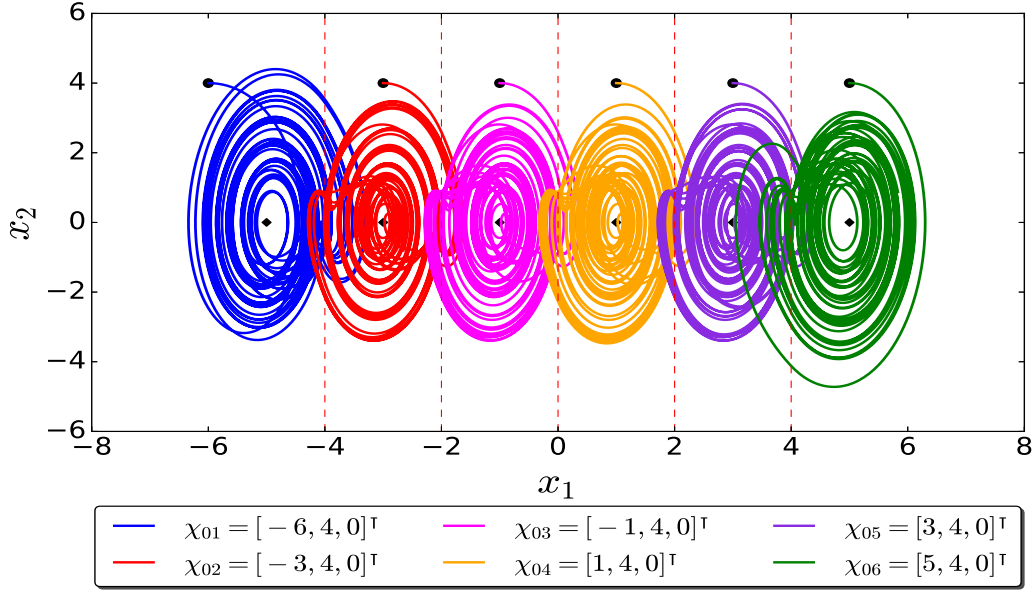}
\caption{Projections of different attractors based on UDS \textit{Type I} into the ($x_{1},x_{2}$) plane with  $\alpha = 1.9727$ and $\beta = 1.2656$, switching law \eqref{eq:switching_law_6} and $k = 7$ . The black dots indicates the initial conditions.} 
\label{fig:Fig5}
\end{figure}
\begin{figure}[t]
\centering
\hspace{-25pt}\includegraphics[width=14cm,height=8cm]{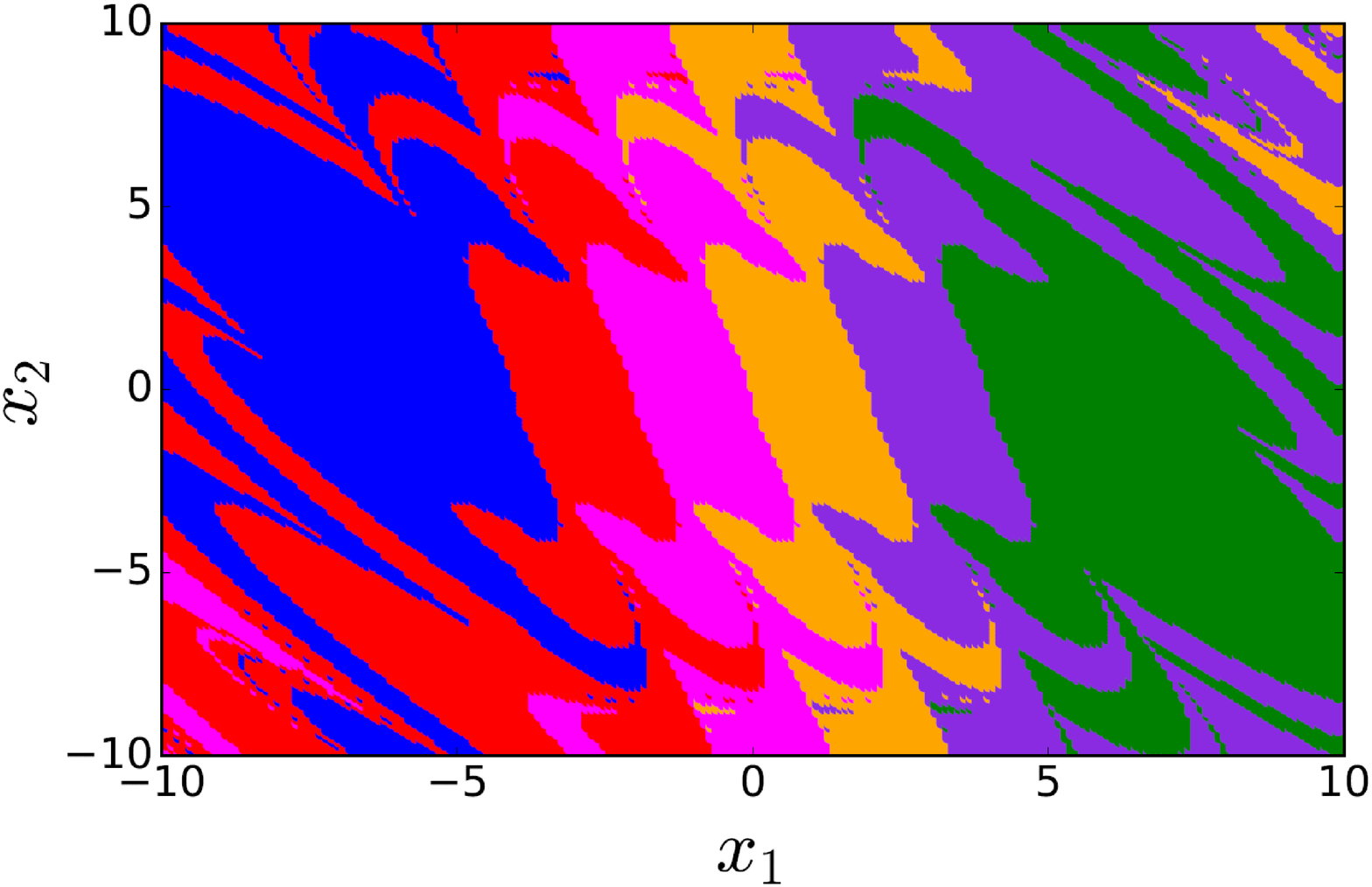}
\caption{Six basins of attraction $\Omega_{i}$, for $i=1,\ldots,6$ by means of using the function \eqref{eq:switching_law_6} and $x_{1} \in [-10,10]$, $x_{2} \in [-10,10]$ and $x_{3}=0$. } 
\label{fig:Fig6}
\end{figure}

\noindent {\bf Example 4:} 
It is worth to mention that it is possible to extend the number of final states of the UDS \textit{Type 1} by increasing the number of  switching domains of the systems. The key idea is to design appropriately the switching law by introducing more domains \cite{Ontanon-Garcia2014}. In order to illustrate how to increase the number of final stable states we change the switching law \eqref{eq:switching_law_3} as follows:
\begin{equation}\label{eq:switching_law_6}
\sigma(\x) =  \left\{
\begin{array}{lll}
   -5,  & \text{if}  & \x \in \hat{S}_{1} = \{\x \in \mathbf{R}^{3}: x_{1}  \leq -4\}; \\
   -3,  & \text{if}  & \x \in \hat{S}_{2} = \{\x \in \mathbf{R}^{3}: -4 < x_{1} \leq -2 \};  \\
   -1,  & \text{if}  & \x \in \hat{S}_{3} = \{\x \in \mathbf{R}^{3}: -2 < x_{1} \leq  0 \};   \\
    1,  & \text{if}  & \x \in \hat{S}_{4} = \{\x \in \mathbf{R}^{3}:  0 <  x_{1} \leq 2 \};   \\
    3,  & \text{if}  & \x \in \hat{S}_{5} = \{\x \in \mathbf{R}^{3}:  2 <  x_{1} \leq 4\};   \\
    5,  & \text{if}  & \x \in \hat{S}_{6} = \{\x \in \mathbf{R}^{3}:   4 <  x_{1} \}.   \\
\end{array} \right .
\end{equation} 

In Figure \eqref{fig:Fig5} we shown the behavior of the dynamical system defined by \eqref{eq:PWL}-\eqref{eq:uds_type1}  with the same parameter values as in our previous example but with the switching law \eqref{eq:switching_law_6}. We show the dynamics of the system with the following initial conditions:  $\chi_{01} = (-6, 4, 0)^{\top}$, $\chi_{02} = (-3, 4, 0)^{\top}$, $\chi_{03} = (-1,4, 0)^{\top}$, $\chi_{04} = (1, 4, 0)^{\top}$, $\chi_{05} = (3,4,0)^{\top}$ and $\chi_{06} = (5,4,0)^{\top}$.

On the other hand, in Figure \eqref {fig:Fig6} we show the shapes of six basin of attractions $\Omega_{i}$ (for $i=1,\ldots,6$) generated with the switching law \eqref{eq:switching_law_6} and we vary the initial condition in the range $x_{1} \in [-10,10]$, $x_{2} \in [-10,10]$ and $x_{3}=0$. 

\section{Emerging multistability in a multiscroll attractor based on UDS \textit{Type II}}

Now, the interest is to generate multistability behavior via a dynamical system based on UDS  \textit{Type II}, so we consider the system \eqref{eq:PWL} with \eqref{eq:uds}.

The idea of generalized multistability generation is different to that given in Section \ref{sec_multistabilityUDSI}, instead of controlling the stable and unstable manifolds, it is increased the number of domains in the partion of the phase space. Recall that the spectra $\Lambda=\{\lambda_1,\lambda_2,\lambda_3\}$ of the linear operator $A$ is given as follows: $0<\lambda_1\in\R$, and $\lambda_2,\lambda_3\in\C$ is a pair of complex conjugate with negative real part and corresponding eigenvectors $\bar{\textit{v}_{j}} \in \mathbf{R}^{n}$, $j=1,2,3$.  Our starting point is example 2 where $\alpha=-0.6$, $\beta=6$ and $\gamma=0.6$, and the phase space is partitioned by $S_1=\{\x\in\R^3|x_1\geq-1\}$ and $S_2=\{\x\in\R^3|x_1<-1\}$. Each domain has a stable manifolf $E^{s}_1\subset S_1$ and $E^{s}_2\subset S_2$ given by planes such that they are parallel $E^{s}_1 \parallel E^{s}_2$ . The basin of attraction is located $\Omega$ between  $E^{s}_1$ and $ E^{s}_2$ . So, now the idea of generalized  multistability generation is by incresasing the number of domains in the partition and generate an attractor near the switching surface and between two stable manifolds, i.e.,   $E^{s}_1\subset S_1$, $E^{s}_2\subset S_2$, $\ldots$, $E^{s}_k\subset S_m$, with $2\leq m\in \Z$, and $E^{s}_1 \parallel E^{s}_2,\ldots, E^{s}_{m-1} \parallel E^{s}_m$.

\noindent {\bf Example 5:} 
We exemplify the multistability based on UDS type II by a PWL system which is capable of producing generalized bistability.   The PWL system given in example 2 is used but now the phase space is partitioned in three domains given by  $ S_{1x_1}=\{\x\in\R^3|x_1 >1\}$, $ S_{2x_1}=\{\x\in\R^3|-1\leq x_1\leq 1\}$ and  $ S_{3x_1}=\{\x\in\R^3| x_1 < -1\}$. The switching surfaces are given by $\Sigma_1=\{\x\in\R^3|x_1=1\}$ and $\Sigma_2=\{\x\in\R^3|x_1=-1\}$. This action of introducing a new domain modifies the function $\sigma(\x)$ given by \eqref{eq:switching_law_3_prima} as follows:

\begin{equation}\label{eq:switching_law_3_prima_2}
\sigma(\x) =  \left\{
\begin{array}{lll}
   -7, & {\text if}  &  \x \in  S_{1x_1};  \\
    0,  & {\text if}  & \x \in  S_{2x_1};  \\
    7,  & {\text if}  & \x \in  S_{3x_1};   \\\end{array} \right .
\end{equation} 

\begin{figure}[t]
\centering
\hspace{-25pt}\includegraphics[width=14cm,height=8cm]{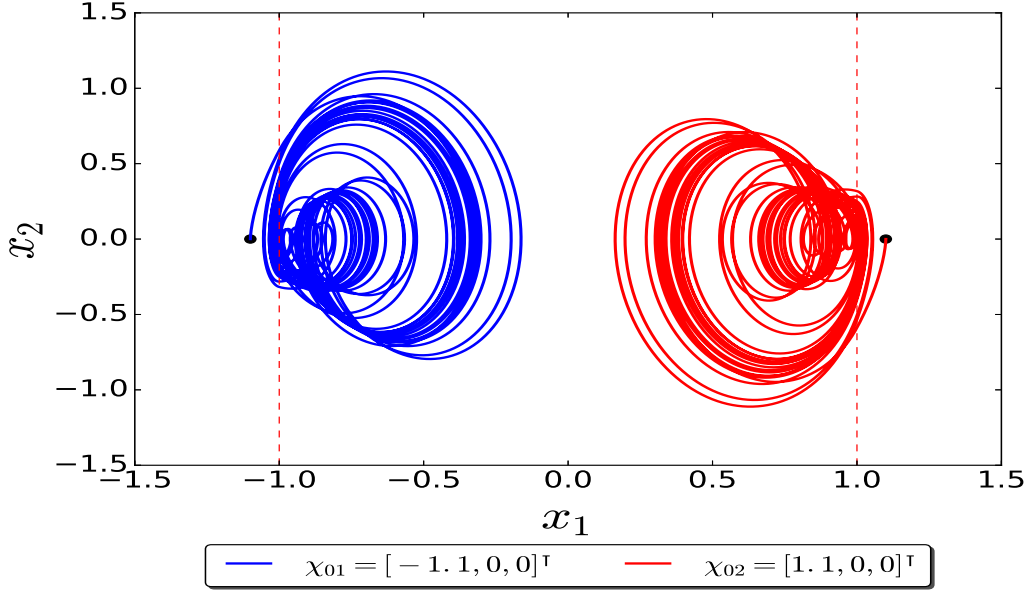}
\caption{Projections of the attractors based on UDS \textit{Type II} onto the ($x_{1},x_{2}$) plane, with the  switching law \eqref{eq:switching_law_3_prima_2} and $\alpha = 0.6$ and $p=10$ . The black dots indicate initial conditions  for the left-hand side attractor $\mathcal{A_L}$ and right-hand side attractor $\mathcal{A_R}$, respectively.} 
\label{fig:Fig7}
\end{figure}

 Now,the equilibria are located at $x_1^*=(0,0,0)$ and $x_{2,3}^*=(\pm 11.66,0,0)$, and this dynamical system presents two attractors $\mathcal{A_{L}}$ and $\mathcal{A_{R}}$, which are shown in Figure \eqref{fig:Fig7}. The left-hand side attractor $\mathcal{A_{L}}$ and right-hand side attractor $\mathcal{A_{R}}$ are generated by considering the following initial conditions: $\chi_{01} = (-1.1, 0, 0)^{\top}$ and $\chi_{02} = (1.1, 0, 0)^{\top}$. In terms of generalized multistability we have a biestable behavior and the basin of attraction of the system is given by the union of two basins of attraction $\Omega_{1} \cup \Omega_{2}$. Figure \eqref{fig:Fig8} shows the basins of attraction $\Omega_{1}$ and $\Omega_{2}$ corresponding to two attractors  $\mathcal{A_{L}}$ and $\mathcal{A_{R}}$, respectively.

The generalization of bistable to multistable behavior given by a dynamical system based on UDS {\it Type II} can be given by adding more domains to the partition by considering swithcing surfaces $\Sigma_i$ perpendicular to the axis $x_1$ based on the aforementioned. But it is not the only way how we can add more domains The other idea  is explained with the next example.
\begin{figure}[t]
\centering
\hspace{-25pt}\includegraphics[width=14cm,height=8cm]{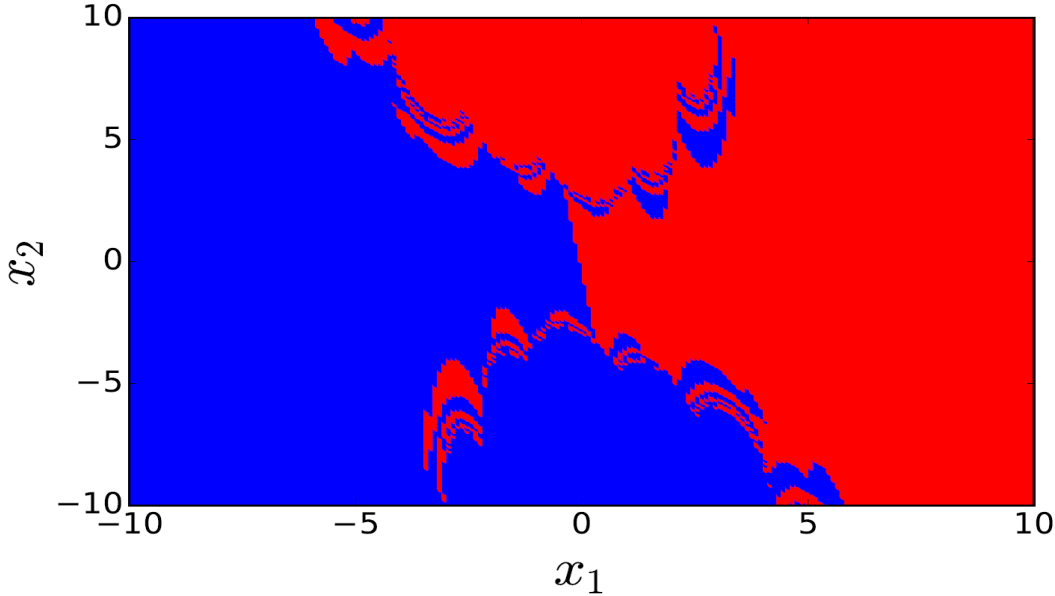}
\caption{The basins of attraction $\Omega_i$, for $i=1,2$, of the system given in example 2 and the switching function
\eqref{eq:switching_law_3_prima_2} in $x_{1} \in [-10,10]$, $x_{2} \in [-10,10]$ and $x_{3}=0$. Blue and red points correspond to initial conditions that converge to the $\mathcal{A_L}$ attractor and the $\mathcal{A_R}$ attractor, respectively.} 
\label{fig:Fig8}
\end{figure}

\begin{figure}[ht]
\centering
\hspace{-25pt}\includegraphics[width=14cm,height=8cm]{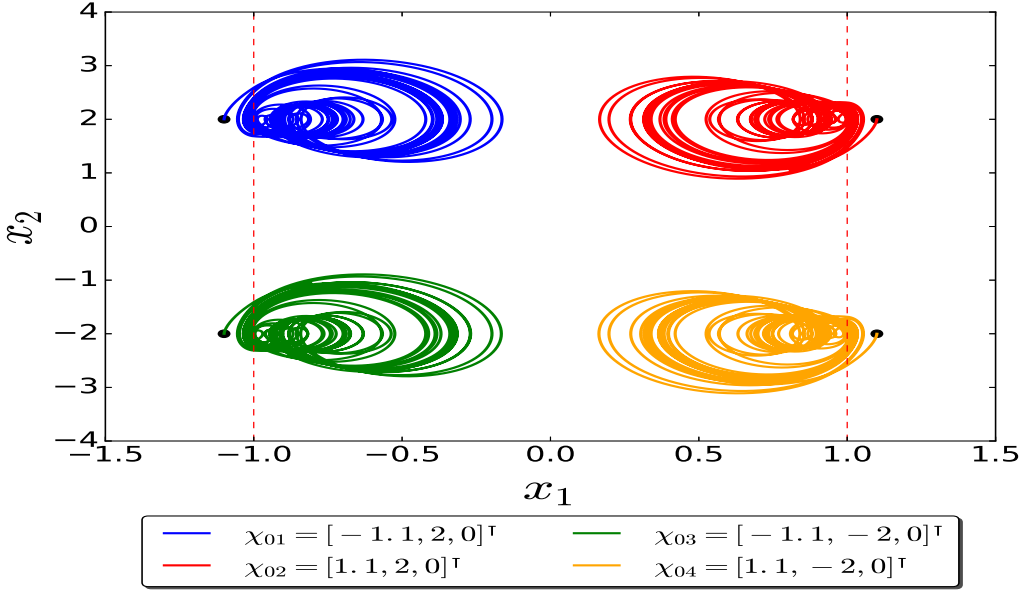}
\caption{Projections of the attractors based on UDS \textit{Type II} onto the ($x_{1},x_{2}$) plane with $\alpha = 0.7$;  $p = 10$ and, vector-valued function $B$ given by \eqref{eq:uds_type2_prima}.} 
\label{fig:Fig9}
\end{figure}

\noindent {\bf Example 6:} 

Now we explain how to generate four  attractors in a two dimensional grid ($2D$-grid scroll attractors) by modifying the piecewise constant vector $B$.  The equilibria of the system is given by $\x^*=(\sigma(\x)/\alpha,0,0)^\top$ by considering \eqref{eq:uds}. Notice that the equilibria are loacted in the axis $x_1$ but now we want the equilibria are located onto the plane $(x_1,x_2)$ as follows $\x^*=(\sigma(\x)/\alpha,f(\x),0)^\top$, thus $B=-A\x$.
  
  \begin{equation}\label{eq:uds_type2_prima}
B(\x) =  \left( \begin{array}{c}
           -f(\x) \\
           0 \\
           \sigma(\x) + \beta f(\x)   
          \end{array} \right), 
\end{equation}

\noindent where  $\sigma(\cdot)$ is the switching law give by \eqref{eq:switching_law_3_prima_2}, and $f(\x)$ is the following step function:
\begin{equation}\label{eq:switching_law_f}
f(\x) =  \left\{
\begin{array}{lll}
   -1.4,  & \text{ if}  & \x \in S_{1x_2} = \{\x \in \mathbf{R}^{3}: x_{2}  \leq 0 \}; \\
    1.4,  & \text{ if}  & \x \in S_{2x_2} = \{\x \in \mathbf{R}^{3}: x_{2}  > 0 \}.  \\
\end{array} \right .
\end{equation} 

The role of the function $f(\cdot)$ is to split the $x_{2}$ direction for each one of the switching surfaces $S_{ix_2}$ and add more domains.  Now the space $\mathbf{R^{3}}$ is partitioned in six domains given as follows: $S_1=S_{1x_1}\cap S_{1x_2}$, $S_2=S_{2x_1}\cap S_{1x_2}$, $S_3=S_{3x_1}\cap S_{1x_2}$, $S_4=S_{1x_1}\cap S_{2x_2}$, $S_5=S_{2x_1}\cap S_{2x_2}$, $S_6=S_{3x_1}\cap S_{2x_2}$. Now there are six equilibria located at: $x_{1,3,4,5}^*=(\pm 11.66,\pm 2,0)$ and $x_{2,5}^*=(0,\pm 2,0)$, with the three equilibrium points added is possible  generate four attractors. Figure \eqref{fig:Fig9} shows multistable behavior for the four coexisting attractors, which are  generated by using the following initial conditions: $\chi_{01} = (-1.1, 2, 0)^{\top}$, $\chi_{02} = (1.1, 2, 0)^{\top}$, $\chi_{03} = (-1.1,- 2, 0)^{\top}$ and $\chi_{04} = (1.1, -2, 0)^{\top}$. Each final stable state of the system is a single chaotic attractor  witch depend only of the initial condition selected. In Figure \eqref{fig:Fig10}, basins of attraction of the system based on UDS \textit{Type II} is shown by varying the initial conditions in the range $x_{1} \in [-10,10]$, $x_{2} \in [-10,10]$ and $x_{3} = 0$. As with UDS \textit{Type I} is possible to extend the number of
final states by adding more equlibria along $x_2$ with switching domains.  

\begin{figure}[ht]
\centering
\hspace{-25pt}\includegraphics[width=14cm,height=8cm]{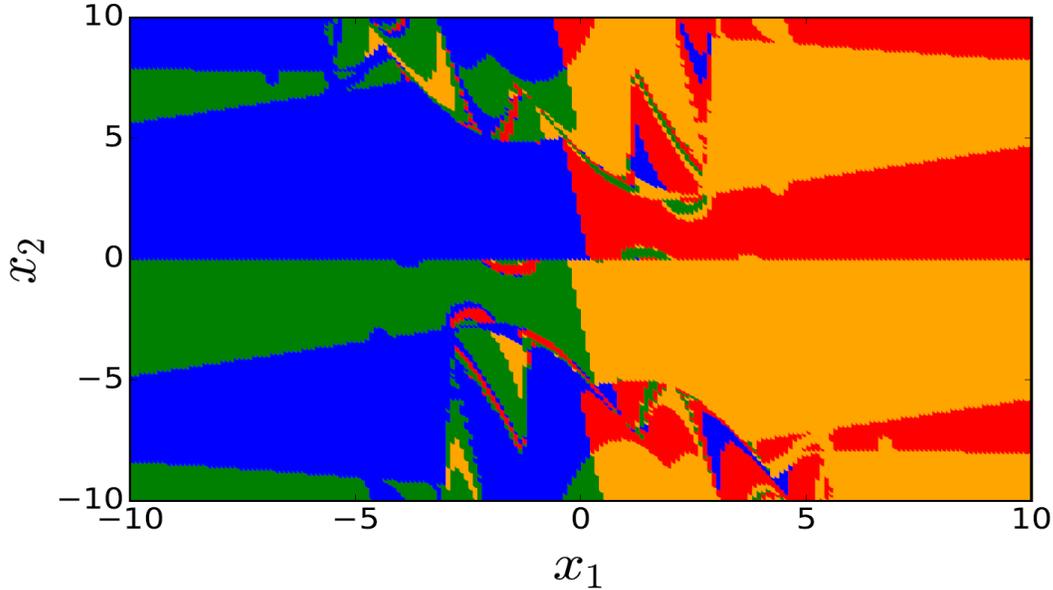}
\caption{Basins of attraction,  of the attractors based on UDS \textit{Type II}, onto the $(x_1,x_2)$ plane, considering $x_{1} \in [-10,10]$, $x_{2} \in [-10,10]$ and $x_{3}=0$. and, vector-valued function $B$ given by \eqref{eq:uds_type2_prima}. Each point represents a given initial condition and its color is the basin of attractions in which the UDS converge with such initial condition. } 
\label{fig:Fig10}
\end{figure}

\section{Concluding remarks}

We have proposed two methodologies to design dynamical systems based on unstable dissipative systems either of \textit{Type I} or \textit{Type II} in such a way that both types of systems generate a multistable behavior. The former methodology consists in introducing a bifurcation parameter in the linear operator of the UDS of \textit{Type I}. With such parameter, we can change the location of the stable and unstable manifold until the trajectory is trapped in a specific attracting set.  In regard to our second methodology, we have considered a UDS \textit{Type II} and more domains have been added to the particion of the phase space. The domains were added by modifing the switching law without changing the linear operator. With this methodology we can design a priory the number of attractors by changing the vector $B$. 

\section*{Acknowledgements}
H. E. Gilardi-Vel\'azquez is a doctoral fellow of the CONACYT in the Graduate Program on control and dynamical systems at DMAp-IPICYT.

%

\end{document}